\documentclass[11pt,reqno]{article}
\usepackage[utf8]{inputenc}
\usepackage{amsmath,amssymb,amsthm}
\usepackage{geometry}
\usepackage{hyperref}
\usepackage{booktabs}
\usepackage{array}
\usepackage{cite}
\usepackage{tikz}
\usepackage{xcolor}
\usetikzlibrary{shapes.geometric, shadows, positioning, backgrounds, calc}

\newtheorem{theorem}{Theorem}[section]
\newtheorem{lemma}[theorem]{Lemma}
\newtheorem{proposition}[theorem]{Proposition}

\theoremstyle{definition}

\newtheorem{example}[theorem]{Example}
\newtheorem{remark}[theorem]{Remark}

\numberwithin{equation}{section}

\title{\textbf{Spectral Characterization and Network Systems Dynamics of Zero-Divisor Topologies over $\mathbb{F}_p[x]/\langle x^4\rangle$: Consensus, Reliability, and Transport}}

\author{
  \textbf{Apurba Sarkar}$^{1}$, \textbf{Kalyan Hansda}$^{1}$, and \textbf{Makhan Maji}$^{2}$\\[1ex]
  $^{1}$Department of Mathematics, Visva-Bharati, Santiniketan -- 731235, West Bengal, India.\\
  $^{2}$Indian Institute of Technology Madras, Chennai -- 600036, Tamil Nadu, India.\\
  \texttt{apurbasarkar065@gmail.com, kalyanh4@gmail.com, makhan2maths@gmail.com}
}
\date{}

\begin{document}
\maketitle

\begin{abstract}
For an odd prime $p$,  let $\Gamma(R)$ denote the zero-divisor graph associated with the finite local ring $R= \mathbb{F}_p[x]/\langle x^4 \rangle$, having $p^3 - 1$ vertices and size $|E(\Gamma(R))| = \frac{1}{2}(3p^4 - 4p^3 - p^2 + 2)$.   Using an equitable 4-cell partition induced by the ideal filtration $\langle u^3 \rangle \subset \langle u^2 \rangle \subset \langle u \rangle$, we prove that $\operatorname{Spec}_A(\Gamma(R))$ decomposes into discrete levels $\lambda \in \{0, -1\}$ with multiplicities $p^3 - p^2 - 1$ and $p^2 - 3$, respectively, alongside three irrational roots of an irreducible cubic polynomial $P_3(\lambda) = 0$, yielding adjacency energy $\epsilon(\Gamma(R)) = 4p^2 + \mathcal{O}(p)$. In contrast, the Laplacian spectrum is entirely integral: $\operatorname{Spec}_L(\Gamma(R)) = \{0^1, (p - 1)^{p^3 - p^2}, (p^2 - 1)^{p^2 - p - 1}, (p^3 - 1)^{p - 1}\}$. Consequently, the tree entropy satisfies $z(\Gamma(R)) = \ln p - \frac{1}{p} + \frac{\ln p}{p} + \mathcal{O}(\frac{\ln p}{p^2})$, with small-failure unreliability bounded by $p^2(p - 1)(1 - q)^{p-1}$. In linear consensus $\dot{x} = -Lx$, error trajectories decouple into three modal timescales $\{ (p^3 - 1)^{-1}, (p^2 - 1)^{-1}, (p - 1)^{-1} \}$ with algebraic connectivity $\alpha(\Gamma(R)) = p - 1$ and steady-state $H_2$ error variance $H_2^2 = \frac{Kf(\Gamma(R))}{2n^2} = \frac{1}{2p} + \mathcal{O}(p^{-3})$. Furthermore, the Kirchhoff index yields a vanishing average resistance distance $\overline{r}(\Gamma(R)) = \frac{2}{p} + \mathcal{O}(p^{-3}) \to 0$, while unbiased Markovian random walks induce an asymptotic traffic concentration vector $\Pi = (\frac{1}{3}, 0, \frac{1}{3}, \frac{1}{3})$ with quadratic per-node core routing stress $\pi_{\mathrm{core}}/\pi_{\mathrm{periph}} \sim p^2 + p + 1$.
\end{abstract}



\section{Introduction}\label{Introduction}

The algebraic study of zero-divisor graphs provides an exact framework for constructing highly symmetric discrete architectures from commutative rings \cite{AndersonLivingston1999, AndersonFrazier2019, Beck1988}. For a commutative ring with unity $R$, let $Z^*(R)$ denote its set of non-zero zero-divisors. Following the canonical formulation by Anderson and Livingston \cite{AndersonLivingston1999}, the zero-divisor graph $\Gamma(R)$ has vertex set $V(\Gamma(R)) = Z^*(R)$, where two distinct vertices $x, y \in Z^*(R)$ are adjacent ($x \sim y$) if and only if $xy = 0$. Structural classifications of $\Gamma(R)$ have demonstrated fundamental interactions between annihilating ideal chains and combinatorial invariants \cite{Annamalai2025, Mukhtar2020, Reddy2020, Redmond2002}.

Concurrently, modern distributed systems engineering and network dynamics rely on spectral graph theory to evaluate multi-agent robotic fleets, continuous-time consensus protocols, and diffusive packet transport \cite{Balakrishnan2012, Bapat2014, Chung1997, OlfatiSaberMurray2004}. The operational efficacy of an interconnection topology is fundamentally characterized by three interdependent physical properties: structural reliability and spanning tree complexity to preserve connectivity under edge failures \cite{Bapat2014}; distributed coordination timescales and algebraic connectivity governing decentralized state agreement \cite{Balakrishnan2012, OlfatiSaberMurray2004}; and multi-path diffusive transport captured through resistance distances \cite{KleinRandic1993}, the Kirchhoff index \cite{XiaoGutman2003}, and topological invariants \cite{Gutman1972, Randic1975}.

Unlike finite local chain rings of lower nilpotency index such as $\mathbb{F}_p + u\mathbb{F}_p + u^2\mathbb{F}_p$ ($u^3 = 0$) \cite{Annamalai2025}, where the zero-divisor graph exhibits rational adjacency spectra, setting $u^4 = 0$ triggers a structural transition. For $R = \mathbb{F}_p[x]/\langle x^4\rangle \cong \mathbb{F}_p + u\mathbb{F}_p + u^2\mathbb{F}_p + u^3\mathbb{F}_p$, the intermediate ideal chain $\langle u^3 \rangle \subset \langle u^2 \rangle \subset \langle u \rangle$ induces an equitable partition whose adjacency characteristic polynomial generates an irreducible cubic factor $P_3(\lambda) \in \mathbb{Q}[\lambda]$ with Galois group $S_3$. On $n = p^3 - 1$ vertices with edge size $|E(\Gamma(R))| = \frac{1}{2}(3p^4 - 4p^3 - p^2 + 2)$, this non-trivial annihilation profile governs well-separated Laplacian modes, establishing an explicit framework to evaluate consensus rates, resistance metrics, and random walk transport in distributed network topologies.


\textbf{In this paper, our primary aim is to establish the complete spectral characterization of the zero-divisor graph $\Gamma(R)$ on $ p^3 - 1$ vertices and evaluate its direct applications across network consensus, structural reliability, and diffusive transport.} In Section~\ref{Sec3}, we establish exact closed-form expressions for primary degree- and distance-based topological descriptors of $\Gamma(R)$ and determine its complete algebraic eigenspaces; specifically, using an equitable four-cell partition of $Z^*(R)$, we resolve the full spectrum and multiplicities of the adjacency matrix $A(\Gamma(R))$ and the Laplacian matrix $L(\Gamma(R))$, showing that the adjacency spectrum reduces to an explicit cubic factor alongside discrete levels, whereas the Laplacian spectrum is entirely integral. 

In Section~\ref{Sec4}, we explore network systems applications across multiple domains. In Section~\ref{Sec4_sub1}, we evaluate the exact spanning tree complexity $\tau(\Gamma(R))$ via the Matrix-Tree Theorem, establish the asymptotic scaling of the tree entropy $z(\Gamma(R))$, and show that the all-terminal unreliability under small failure probabilities is strictly governed by the peripheral vertex tier. Furthermore, using the spectral formulation of Xiao and Gutman \cite{XiaoGutman2003}, we evaluate the Kirchhoff index $Kf(\Gamma(R))$ and prove that the average two-point resistance distance satisfies $\overline{r}(\Gamma(R)) \sim \frac{2}{p} \to 0$ as $p \to \infty$, demonstrating that the algebraic core acts as an ultra-low-impedance crossbar. In Section~\ref{Sec4_sub2}, we analyze continuous-time linear consensus dynamics \cite{OlfatiSaberMurray2004}, revealing modal decoupling into three distinct relaxation timescales, and compute the steady-state error variance under additive Wiener noise via the first-order network $H_2$-norm. We then examine Markovian packet routing over $\Gamma(R)$, establishing that stationary traffic mass asymptotically splits equally across three primary partitions and induces a quadratic routing stress ratio $\Theta(p^2)$ on the core nodes.

\section{Preliminaries}\label{Sec2_Preliminaries}

Throughout this paper, let $p$ be an odd prime. We denote by $R = \mathbb{F}_p[x]/\langle x^4 \rangle \cong \mathbb{F}_p + u\mathbb{F}_p + u^2\mathbb{F}_p + u^3\mathbb{F}_p$ ($u^4 = 0$) the finite local commutative chain ring of characteristic $p$ with order $|R| = p^4$. Every element $a \in R$ is uniquely represented as $a = a_0 + a_1 u + a_2 u^2 + a_3 u^3$ with $a_i \in \mathbb{F}_p$. The unique maximal ideal of $R$ is $\mathfrak{m} = \langle u \rangle = uR$, and the set of non-zero zero-divisors is $Z^*(R) = \mathfrak{m} \setminus \{0\}$, having cardinality $|Z^*(R)| = p^3 - 1$. Following Anderson and Livingston \cite{AndersonLivingston1999}, the zero-divisor graph $\Gamma(R)$ has vertex set $V(\Gamma(R)) = Z^*(R)$, where distinct vertices $x, y$ are adjacent ($x \sim y$) if and only if $xy = 0$.

\subsection{Structural Classification of $\Gamma(R)$}

Because annihilation in $R$ is governed by the nilpotency index $u^4 = 0$, the principal annihilator ideals form the strict filtration: $\mathrm{Ann}(u) = \langle u^3 \rangle = u^3 R, \quad \mathrm{Ann}(u^2) = \langle u^2 \rangle = u^2 R, \quad \mathrm{Ann}(u^3) = \langle u \rangle = u R = \mathfrak{m}$. This induces a natural partition of $V(\Gamma(R))$ into three degree-invariant classes $V(\Gamma(R)) = \mathcal{D}_1 \cup \mathcal{D}_2 \cup \mathcal{D}_3$:
\begin{enumerate}
    \item $\mathcal{D}_1 = uR \setminus u^2R = \{x \in Z^*(R) : \mathrm{Ann}(x) = u^3R\}$, with cardinality $N_1 = p^2(p - 1)$ and uniform degree $d_1 = |\mathcal{D}_3| = p - 1$. The induced subgraph $\Gamma[\mathcal{D}_1]$ is an independent set.
    \item $\mathcal{D}_2 = u^2R \setminus u^3R = \{x \in Z^*(R) : \mathrm{Ann}(x) = u^2R\}$, with cardinality $N_2 = p(p - 1)$ and uniform degree $d_2 = (N_2 - 1) + N_3 = p^2 - 2$. The induced subgraph $\Gamma[\mathcal{D}_2]$ is a complete graph $K_{p(p-1)}$.
    \item $\mathcal{D}_3 = u^3R \setminus \{0\} = \{x \in Z^*(R) : \mathrm{Ann}(x) = uR\}$, with cardinality $N_3 = p - 1$ and uniform degree $d_3 = |V(\Gamma(R))| - 1 = p^3 - 2$. The induced subgraph $\Gamma[\mathcal{D}_3]$ is a complete graph $K_{p-1}$ whose vertices are adjacent to every other vertex in $\Gamma(R)$.
\end{enumerate}

The edge set of $\Gamma(R)$ naturally decomposes into four mutually disjoint classes determined by the degrees of their endpoints across the tiers $\mathcal{D}_1, \mathcal{D}_2, \mathcal{D}_3$:
\begin{itemize}
    \item []$E_1 = E(\mathcal{D}_1, \mathcal{D}_3)$ with cardinality $|E_1| = |\mathcal{D}_1||\mathcal{D}_3| = p^2(p-1)^2$;
    \item []$E_2 = E(\Gamma[\mathcal{D}_2])$ with cardinality $|E_2| = \binom{|\mathcal{D}_2|}{2} = \frac{1}{2}p(p-1)(p^2 - p - 1)$;
    \item []$E_3 = E(\mathcal{D}_2, \mathcal{D}_3)$ with cardinality $|E_3| = |\mathcal{D}_2||\mathcal{D}_3| = p(p-1)^2$;
    \item []$E_4 = E(\Gamma[\mathcal{D}_3])$ with cardinality $|E_4| = \binom{|\mathcal{D}_3|}{2} = \frac{1}{2}(p-1)(p-2)$.
\end{itemize}
Consequently, the total edge size is $|E(\Gamma(R))| = |E_1| + |E_2| + |E_3| + |E_4| = \frac{1}{2}(3p^4 - 4p^3 - p^2 + 2)$, and the average vertex degree is $\overline{d} = \frac{2|E(\Gamma(R))|}{|V(\Gamma(R))|}$.

\begin{remark}
\label{rem:graph_family_comparison}
To situate $\Gamma(R)$ within structural graph theory, we compare it against standard graph classes:
\begin{enumerate}
    \item  A graph is \textit{complete multipartite} \cite{GodsilRoyle2001} if and only if non-adjacency is an equivalence relation. In $\Gamma(R)$, selecting $v_1 \in \mathcal{D}_1$, $v_2 \in \mathcal{D}_2$, and $w_1 \in \mathcal{D}_1 \setminus \{v_1\}$, we observe $v_1 \not\sim w_1$ and $v_1 \not\sim v_2$, yet $v_2 \sim w_2$ for any $w_2 \in \mathcal{D}_2 \setminus \{v_2\}$. Thus non-adjacency is non-transitive. Moreover, the block $\mathcal{D}_2$ contains internal edges ($d_{\mathcal{D}_2}(v) = p(p - 1) - 1 > 0$), violating the independent set property required of partite sets.
    \item  A \textit{strongly regular} \cite{GodsilRoyle2001} graph must be regular. However, $\Gamma(R)$ exhibits three distinct vertex degrees: $d(v) \in \{p - 1,\, p^2 - 2,\, p^3 - 2\}$, precluding regularity.
    \item  Because \textit{Cayley graphs} \cite{GodsilRoyle2001} are vertex-transitive, all vertices share identical local neighborhoods. The strict degree hierarchy $d(\mathcal{D}_1) < d(\mathcal{D}_2) < d(\mathcal{D}_3)$ rules out vertex-transitivity and Cayley representation.
    \item Consequently, $\Gamma(R)$ is a \textit{multi-tiered core-periphery graph} generated by the nilpotency chain $\langle u^3 \rangle \subset \langle u^2 \rangle \subset \langle u \rangle$. The set $\mathcal{D}_3$ functions as an all-dominating clique core, $\mathcal{D}_1$ forms an independent peripheral shell whose connections terminate entirely in $\mathcal{D}_3$, and $\mathcal{D}_2$ acts as a dense, interconnecting intermediate layer.
\end{enumerate}
\end{remark}

\subsection{Equitable Partitions and Quotient Matrices}

To analyze the spectra of structured networks, we utilize equitable partitions. A partition $\mathcal{P} = \{V_1, \dots, V_k\}$ of $V(\Gamma)$ is equitable if every vertex in $V_i$ has exactly $q_{ij} = |N(v) \cap V_j|$ neighbors in $V_j$ for all $1 \le i, j \le k$. 

Let $S \in \{0, 1\}^{n \times k}$ denote the unnormalized block indicator matrix of $\mathcal{P}$, satisfying $S^T S = D_\mathcal{P} = \operatorname{diag}(|V_1|, \dots, |V_k|)$. In algebraic graph theory \cite{Bapat2014, BrouwerHaemers2012, GodsilRoyle2001}, the standard equitable quotient matrix $Q_A = (q_{ij})_{k \times k}$ satisfies the intertwining relation $A(\Gamma) S = S Q_A$.
The characteristic polynomial of $Q_A$ divides that of $A(\Gamma)$. The eigenvalues of $Q_A$ correspond to eigenvectors of $A(\Gamma)$ that are constant on each partition cell $V_i$, while the remaining $|V| - k$ eigenvalues correspond to eigenvectors orthogonal to the all-ones vector on each cell \cite{Balakrishnan2012, Bapat2014}. 

Although $Q_A$ is non-symmetric whenever $|V_i| \neq |V_j|$, it satisfies the combinatorial balance condition $|V_i| q_{ij} = |V_j| q_{ji} = e(V_i, V_j)$. Consequently, $Q_A$ is diagonally similar via $D_\mathcal{P}^{1/2} Q_A D_\mathcal{P}^{-1/2}$ to the real symmetric matrix $S_A = D_\mathcal{P}^{-1/2} S^T A(\Gamma) S D_\mathcal{P}^{-1/2}$, guaranteeing that all eigenvalues of $Q_A$ are strictly real \cite{GodsilRoyle2001}.

\subsection{Algebraic Graph Spectra, Electrical Resistance, and Consensus Dynamics}

Let $\Gamma = (V, E)$ be a connected, undirected graph of order $n$ with adjacency matrix $A(\Gamma)$ and diagonal degree matrix $D(\Gamma)$. The combinatorial Laplacian matrix is $L(\Gamma) = D(\Gamma) - A(\Gamma)$ \cite{Bapat2014, Chung1997}. Its eigenvalues are ordered as $0 = \mu_0 < \mu_1 \le \dots \le \mu_{n-1}$, where $\mu_1 = \lambda_2(L(\Gamma)) = \alpha(\Gamma)$ denotes the algebraic connectivity governing the convergence rate of continuous-time consensus protocols $\dot{x}(t) = -L(\Gamma)x(t)$ \cite{OlfatiSaberMurray2004}. By the Matrix-Tree Theorem, the number of spanning trees is $\tau(\Gamma) = \frac{1}{n}\prod_{i=1}^{n-1}\mu_i$, with asymptotic tree entropy $z(\Gamma) = \frac{1}{n}\ln \tau(\Gamma)$.

The effective electrical resistance between vertices $u$ and $v$ is given by $r(u, v) = (e_u - e_v)^T L^\dagger (\Gamma) (e_u - e_v)$ \cite{KleinRandic1993}, and the Kirchhoff index is defined as \cite{XiaoGutman2003}:
\begin{equation}
Kf(\Gamma) = \sum_{u < v} r(u, v) = n \sum_{i=1}^{n-1} \frac{1}{\mu_i}. \label{eq:kirchhoff_def}
\end{equation}
The average resistance distance is $\overline{r}(\Gamma) = \frac{2}{n(n - 1)} Kf(\Gamma)$. In distributed multi-agent systems driven by standard white noise, $dx(t) = -L(\Gamma)x(t)dt + dW(t)$, the steady-state state-error variance around the consensus state is quantified by the first-order network $H_2$-norm.
\begin{equation}
H_2^2(\Gamma) = \lim_{t \to \infty} \frac{1}{n} \mathbb{E}\left[ \|x(t) - \overline{x}(t)\mathbf{1}_n\|^2 \right] = \frac{1}{2n}\sum_{i=1}^{n-1}\frac{1}{\mu_i} = \frac{Kf(\Gamma)}{2n^2}. \label{eq:h2_norm_def}
\end{equation}

\begin{remark}[Universality of the $H_2$--Kirchhoff Identity]
\label{rem:h2_universality}
The identity $H_2^2(G) = \frac{Kf(G)}{2n^2}$ in \eqref{eq:h2_norm_def} is a universal exact identity holding for any finite, connected, and undirected graph $G$ on $n$ vertices \cite{Bamieh2012, XiaoBoydKim2007}. On the zero-sum subspace $\mathbf{1}_n^\perp$, the steady-state error covariance $\Sigma = \lim_{t \to \infty}\mathbb{E}[\delta(t)\delta(t)^T]$ satisfies the Lyapunov equation $L\Sigma + \Sigma L = \Pi$, where $\Pi = I_n - \frac{1}{n}\mathbf{1}_n\mathbf{1}_n^T$. Diagonalization yields $\Sigma = \frac{1}{2} L^\dagger$, whence $H_2^2(G) = \frac{1}{n}\operatorname{tr}(\Sigma) = \frac{1}{2n}\operatorname{tr}(L^\dagger) = \frac{1}{2n}\sum_{i=1}^{n-1}\mu_i^{-1} = \frac{Kf(G)}{2n^2}$ follows immediately.
\end{remark}

\subsection{Diffusive Transport and Degree-Based Descriptors}

For discrete-time diffusive transport, an unbiased random walk on $\Gamma$ is governed by the transition matrix $P = D^{-1}(\Gamma)A(\Gamma)$, with stationary distribution $\pi_v = \frac{d(v)}{2|E(\Gamma)|}$. Finally, structural and branching complexity are measured using standard topological descriptors:
\begin{itemize}
    \item[]\textit{First and second Zagreb indices} \cite{Gutman1972}: $M_1(\Gamma) = \sum_{v \in V} (d(v))^2$ and $M_2(\Gamma) = \sum_{uv \in E} d(u)d(v)$.
    \item []\textit{Randi\'{c} index} \cite{Randic1975}: $R(\Gamma) = \sum_{uv \in E} (d(u)d(v))^{-1/2}$.
    \item [] \textit{Wiener index }\cite{Wiener1947}: $W(\Gamma) = \sum_{u < v} d(u, v)$.
    \item [] \textit{Forgotten and Hyper-Zagreb indices} \cite{FurtulaGutman2015, Shirdel2013}: $F(\Gamma) = \sum_{v \in V} (d(v))^3$ and $HM(\Gamma) = \sum_{uv \in E} (d(u) + d(v))^2 = F(\Gamma) + 2M_2(\Gamma)$.
\end{itemize}

\section{Topological Invariants and Complete Spectral Decomposition of $\Gamma(R)$}\label{Sec3}

In this section, we establish the structural and algebraic matrix characterizations of the zero-divisor graph $\Gamma(R)$. Leveraging the annihilator-based equitable partition of $Z^*(R)$, we first evaluate its degree- and distance-based topological descriptors and subsequently determine the full eigenspaces and multiplicities of both its adjacency and Laplacian matrices.

\begin{proposition}

Let $p$ be an odd prime. Then the following assertions hold for a  zero-divisor graph $\Gamma(R)$:
\begin{enumerate}
    \item The first Zagreb index is $M_1(\Gamma(R)) = (p-1)\left[p^6 + p^5 + p^4 - 10p^3 + p^2 + 4p + 4\right]$.

   \item The second Zagreb index is $ M_2(\Gamma(R)) = \frac{1}{2}(p-1)\left[6p^7 - 9p^6 - 7p^5 - 4p^4 + 28p^3 - 8p - 8\right]$.

    \item The Randi\'{c} index is
    $
    R(\Gamma(R)) = \frac{p^2(p-1)\sqrt{p-1}}{\sqrt{p^3 - 2}} + \frac{p(p-1)(p^2 - p - 1)}{2(p^2 - 2)} + \frac{p(p-1)^2}{\sqrt{(p^2 - 2)(p^3 - 2)}} + \frac{(p-1)(p-2)}{2(p^3 - 2)}$.

    \item The Wiener index is
    $W(\Gamma(R)) = \frac{1}{2}(p-1)\left[2p^5 + 2p^4 - p^3 - 3p^2 - 2p - 2\right]$.
  
\end{enumerate}
\end{proposition}
\begin{proof}
(1): By definition, $M_1(\Gamma(R)) = \sum_{v \in V(\Gamma(R))} (d(v))^2$. Partitioning over the degree classes $\mathcal{D}_1, \mathcal{D}_2, \mathcal{D}_3$ we have 
\begin{align*}
M_1(\Gamma(R)) &= N_1 d_1^2 + N_2 d_2^2 + N_3 d_3^2 \\
&= p^2(p-1)(p-1)^2 + p(p-1)(p^2 - 2)^2 + (p-1)(p^3 - 2)^2 \\
&= (p-1)\left[p^2(p-1)^2 + p(p^2 - 2)^2 + (p^3 - 2)^2\right]\\
&=(p-1)[p^2(p^2 - 2p + 1) + p(p^4 - 4p^2 + 4) + (p^6 - 4p^3 + 4)]\\
&= (p-1)\left[(p^4 - 2p^3 + p^2) + (p^5 - 4p^3 + 4p) + (p^6 - 4p^3 + 4)\right] \\
&= (p-1)(p^6 + p^5 + p^4 - 10p^3 + p^2 + 4p + 4).
\end{align*}

(2): The second Zagreb index of $\Gamma(R)$ is given by
\begin{align*}
M_2(\Gamma(R)) &= \sum_{uv \in E(\Gamma(R))} d(u)d(v) \\ 
&= |E_1| d_1 d_3 + |E_2| d_2^2 + |E_3| d_2 d_3 + |E_4| d_3^2 \\
&= p^2(p-1)^2 (p-1)(p^3 - 2) + \frac{1}{2}p(p-1)(p^2 - p - 1)(p^2 - 2)^2 \\
&\quad + p(p-1)^2(p^2 - 2)(p^3 - 2) + \frac{1}{2}(p-1)(p-2)(p^3 - 2)^2 \\
&= \frac{1}{2}(p-1)\left[6p^7 - 9p^6 - 7p^5 - 4p^4 + 28p^3 - 8p - 8\right].
\end{align*}

(3):  By definition it yields that  
\begin{align*}
R(\Gamma(R)) &=\sum_{uv \in E(\Gamma(R))} \frac{1}{\sqrt{d(u)d(v)}}\\ &= \frac{|E_1|}{\sqrt{d_1 d_3}} + \frac{|E_2|}{d_2} + \frac{|E_3|}{\sqrt{d_2 d_3}} + \frac{|E_4|}{d_3} \\
&= \frac{p^2(p-1)^2}{\sqrt{(p-1)(p^3 - 2)}} + \frac{\frac{1}{2}p(p-1)(p^2 - p - 1)}{p^2 - 2} + \frac{p(p-1)^2}{\sqrt{(p^2 - 2)(p^3 - 2)}} + \frac{\frac{1}{2}(p-1)(p-2)}{p^3 - 2}\\
&=\frac{p^2(p-1)\sqrt{p-1}}{\sqrt{p^3 - 2}} + \frac{p(p-1)(p^2 - p - 1)}{2(p^2 - 2)} + \frac{p(p-1)^2}{\sqrt{(p^2 - 2)(p^3 - 2)}} + \frac{(p-1)(p-2)}{2(p^3 - 2)}.
\end{align*}

(4): Since every vertex in $Z^*(R)$ is adjacent to the universal clique $\mathcal{D}_3$, the distance between any two distinct non-adjacent vertices is 2. Thus, the diameter of $\Gamma(R)$ is 2. The Wiener index is given by 
\begin{align*}
W(\Gamma(R)) &= \sum_{\{u, v\} \subseteq V} d(u, v)\\
&= 1 \cdot |E(\Gamma(R))| + 2 \cdot \left(\binom{|V(\Gamma(R))|}{2} - |E(\Gamma(R))|\right)\\
&= |V|(|V| - 1) - |E|.
\end{align*}
 Now using $|V| = p^3 - 1$ and $|E| = \frac{1}{2}(3p^4 - 4p^3 - p^2 + 2)$ we have 
\begin{align*}
W(\Gamma(R)) &= (p^3 - 1)(p^3 - 2) - \frac{1}{2}(3p^4 - 4p^3 - p^2 + 2) \\
&= (p^6 - 3p^3 + 2) - \frac{3}{2}p^4 + 2p^3 + \frac{1}{2}p^2 - 1 \\
&= p^6 - \frac{3}{2}p^4 - p^3 + \frac{1}{2}p^2 + 1 \\
&= \frac{1}{2}\left(2p^6 - 3p^4 - 2p^3 + p^2 + 2\right)\\
&= \frac{1}{2}(p-1)\left[2p^5 + 2p^4 - p^3 - 3p^2 - 2p - 2\right].
\end{align*}
\end{proof}

Beyond standard indices, the Hyper-Zagreb index $HM(\Gamma) = \sum_{uv \in E} (d(u) + d(v))^2$ and the Forgotten topological index $F(\Gamma) = \sum_{v \in V} (d(v))^3$ capture higher-order degree variance and irregular network branching.

\begin{theorem} \label{thm:forgotten_hyper_zagreb}
Let $p$ be an odd prime. Then the following assertions hold for the zero-divisor graph $\Gamma(R)$:
\begin{enumerate}
    \item The Forgotten topological index of $\Gamma(R)$ is 
    \[
    F(\Gamma(R)) = (p - 1)\left[p^9 + p^7 - 6p^6 - 5p^5 - 3p^4 + 27p^3 - p^2 - 8p - 8\right].
    \]
    \item The Hyper-Zagreb index of $\Gamma(R)$ is 
    \[
    \mathrm{HM}(\Gamma(R)) = (p - 1)\left[p^9 + 7p^7 - 15p^6 - 12p^5 - 7p^4 + 55p^3 - p^2 - 16p - 16\right].
    \]
\end{enumerate}
\end{theorem}

\begin{proof}
(1): By definition, the Forgotten topological index is given by $F(\Gamma(R)) = \sum_{v \in V(\Gamma(R))} (d(v))^3$. Now 
\begin{align*}
F(\Gamma(R)) &= N_1 d_1^3 + N_2 d_2^3 + N_3 d_3^3 \\
&= p^2(p - 1)(p - 1)^3 + p(p - 1)(p^2 - 2)^3 + (p - 1)(p^3 - 2)^3 \\
&= (p - 1)\left[p^2(p - 1)^3 + p(p^2 - 2)^3 + (p^3 - 2)^3\right] \\
&= (p - 1)\left[(p^5 - 3p^4 + 3p^3 - p^2) + (p^7 - 6p^5 + 12p^3 - 8p) + (p^9 - 6p^6 + 12p^3 - 8)\right] \\
&= (p - 1)\left[p^9 + p^7 - 6p^6 + (p^5 - 6p^5) - 3p^4 + (3p^3 + 12p^3 + 12p^3) - p^2 - 8p - 8\right] \\
&= (p - 1)\left[p^9 + p^7 - 6p^6 - 5p^5 - 3p^4 + 27p^3 - p^2 - 8p - 8\right].
\end{align*}

(2): We have  $HM(\Gamma(R)) = \sum_{uv \in E} (d(u) + d(v))^2 = \sum_{uv \in E} (d(u)^2 + d(v)^2) + 2\sum_{uv \in E} d(u)d(v)$.
Since $\sum_{uv \in E} (d(u)^2 + d(v)^2) = \sum_{v \in V} d(v) \cdot d(v)^2 = \sum_{v \in V} (d(v))^3 = F(\Gamma(R))$, the identity follows directly.
\end{proof}

We now study the spectral properties of $\Gamma(R)$ by determining the complete eigenspaces and multiplicities of its adjacency and Laplacian matrices via equitable quotient reduction.
Let $A(\Gamma(R))$ and $L(\Gamma(R)) = D(\Gamma(R)) - A(\Gamma(R))$ denote the adjacency and Laplacian matrices of $\Gamma(R)$. We refine the vertex partition into four blocks $V(\Gamma(R)) = V_1 \cup V_2 \cup V_3 \cup V_4$,
\begin{center}
\begin{itemize}
    \item[] $V_1 = \mathcal{D}_1$, with $|V_1| = p^2(p-1)$,
    \item[] $V_2 = S_{u^2}=\{x u^2 : x \in \mathbb{F}_p^*\}$, with $|V_2| = p-1$,
    \item [] $V_3 = \mathcal{D}_3 = S_{u^3}= \{x u^3 : x \in \mathbb{F}_p^*\}$, with $|V_3| = p-1$,
    \item [] $V_4 = S_{u^2+u^3}=\{x u^2 + y u^3 : x, y \in \mathbb{F}_p^*\}$, with $|V_4| = (p-1)^2$.
\end{itemize}
\end{center}
Note that $V_2 \cup V_4 = \mathcal{D}_2$. This 4-block specification forms an equitable partition with quotient representations.

We now establish two fundamental results characterizing the spectral structure and energy of the adjacency matrix $A(\Gamma(R))$. First, we construct the exact eigenspace decomposition of $A(\Gamma(R))$ via an equitable four-cell quotient matrix, and subsequently derive a semi-analytical evaluation and tight asymptotic bounds for its adjacency energy $\epsilon(\Gamma(R))$.

\begin{lemma} \label{lem:equitable_verification}
The four-cell partition $\mathcal{P} = \{V_1, V_2, V_3, V_4\}$ of $V(\Gamma(R))$ is equitable. Specifically, for each cell $V_i$, the number of neighbors $q_{ij} = |N(v) \cap V_j|$ of an arbitrary vertex $v \in V_i$ in cell $V_j$ depends solely on the pair $(i, j)$ and is given by the quotient matrix:
\begin{equation}
Q_A = \begin{pmatrix}
0 & 0 & p - 1 & 0 \\
0 & p - 2 & p - 1 & (p - 1)^2 \\
p^2(p - 1) & p - 1 & p - 2 & (p - 1)^2 \\
0 & p - 1 & p - 1 & p^2 - 2p
\end{pmatrix}.
\end{equation}
\end{lemma}

\begin{proof}
Let $a \in Z^*(R)$ be uniquely expanded as $a = a_1 u + a_2 u^2 + a_3 u^3$ with $a_i \in \mathbb{F}_p$. Recall that the cells are defined ring-theoretically by:
$V_1 = \{a \in Z^*(R) : a_1 \neq 0\}, 
\;V_2 = \{a_2 u^2 : a_2 \in \mathbb{F}_p^\times\}, 
V_3 = \{a_3 u^3 : a_3 \in \mathbb{F}_p^\times\}, \textrm{and }
\;V_4 = \{a_2 u^2 + a_3 u^3 : a_2, a_3 \in \mathbb{F}_p^\times\}$.
with $ |V_1| = p^2(p - 1), \;|V_2| = p - 1, \;|V_3| = p - 1, \;\textrm{and } \;|V_4| = (p - 1)^2$.

Note that $V_2 \cup V_4 = \mathcal{D}_2$ and $V_3 = \mathcal{D}_3$. Two distinct vertices $x, y \in Z^*(R)$ are adjacent if and only if $xy = 0$ in $R$, which occurs if and only if $u^4 \mid xy$ in $\mathbb{F}_p[x]$.

We now verify the neighborhood sizes $q_{ij} = |N(v) \cap V_j|$ systematically for each $i \in \{1, 2, 3, 4\}$:

\begin{enumerate}
    \item Neighborhoods of $v \in V_1$: Let $v = a_1 u + a_2 u^2 + a_3 u^3$ with $a_1 \neq 0$. For any $w = b_1 u + b_2 u^2 + b_3 u^3 \in Z^*(R)$, the product is:
    \begin{align*}
    vw &= a_1 b_1 u^2 + (a_1 b_2 + a_2 b_1) u^3 + (a_1 b_3 + a_2 b_2 + a_3 b_1) u^4\\
       &= a_1 b_1 u^2 + (a_1 b_2 + a_2 b_1) u^3.
    \end{align*}
    Thus, $vw = 0$ requires $a_1 b_1 = 0$ implies $b_1 = 0$ (since $a_1 \neq 0$). Substituting $b_1 = 0$ into the $u^3$ term gives $a_1 b_2 = 0$  which implies $b_2 = 0$. Hence, $vw = 0$ if and only if $w \in u^3 R \setminus \{0\} = V_3$. 
    Consequently, for each $j \in \{1, 2, 4\}$, every vertex $w \in V_j$ satisfies either $b_1 \neq 0$ or $b_2 \neq 0$, which implies $vw \neq 0$ and thus $q_{11} = q_{12} = q_{14} = 0$. In contrast, every $w \in V_3$ satisfies $b_1 = b_2 = 0$ with $b_3 \neq 0$, yielding $vw = a_1 b_3 u^4 = 0$; since $v \notin V_3$, the vertices $v$ and $w$ are distinct, so $v$ is adjacent to every vertex in $V_3$, giving $q_{13} = |V_3| = p - 1$.

    \item Neighborhoods of $v \in V_2$: Let $v = a_2 u^2$ with $a_2 \neq 0$. For $w = b_1 u + b_2 u^2 + b_3 u^3$, we have: $vw = a_2 b_1 u^3 + a_2 b_2 u^4 = a_2 b_1 u^3$.  Thus, $vw = 0$ if and only if $a_2 b_1 = 0$ if and only if $ b_1 = 0$, which is equivalent to $w \in u^2 R \setminus \{0\} = V_2 \cup V_3 \cup V_4$.
    
    Since every $w \in V_1$ satisfies $b_1 \neq 0$, no vertex in $V_1$ annihilates $v$, whence $q_{21} = 0$. For any $w \in V_2 \setminus \{v\}$, we have $w = b_2 u^2$ with $b_2 \neq 0$, which gives $vw = a_2 b_2 u^4 = 0$, so $v$ is adjacent to all other vertices in $V_2$ and $q_{22} = |V_2| - 1 = p - 2$. For any $w \in V_3$, the product evaluates to $vw = a_2 b_3 u^5 = 0$, yielding $q_{23} = |V_3| = p - 1$. Finally, every $w \in V_4$ has the form $w = b_2 u^2 + b_3 u^3$ with $b_1 = 0$, ensuring $vw = a_2 b_2 u^4 + a_2 b_3 u^5 = 0$; therefore, $v$ is adjacent to every element in $V_4$, which gives $q_{24} = |V_4| = (p - 1)^2$.
    
    \item Neighborhoods of $v \in V_3$: Let $v = a_3 u^3$ with $a_3 \neq 0$. For any $w \in Z^*(R)$, $w$ has the form $b_1 u + b_2 u^2 + b_3 u^3$ with at least one $b_k \neq 0$. The product is $vw = a_3 b_1 u^4 + a_3 b_2 u^5 + a_3 b_3 u^6 = 0$. Hence, $v$ is adjacent to every distinct vertex in $\Gamma(R)$.
    
    For each $j \in \{1, 2, 4\}$, since $v \notin V_j$, the vertex $v$ connects to every node in $V_j$, which yields $q_{31} = |V_1| = p^2(p - 1)$, $q_{32} = |V_2| = p - 1$, and $q_{34} = |V_4| = (p - 1)^2$. For $j = 3$, $v$ connects to all distinct vertices in $V_3$, giving $q_{33} = |V_3| - 1 = p - 2$.
    
    \item Neighborhoods of $v \in V_4$: Let $v = a_2 u^2 + a_3 u^3$ with $a_2, a_3 \neq 0$. For $w = b_1 u + b_2 u^2 + b_3 u^3$, the product evaluates to $vw = a_2 b_1 u^3$.
    Hence, $vw = 0$ if and only if $b_1 = 0$, exactly identical to the annihilation condition for $V_2$.
For any $w \in V_1$, the condition $b_1 \neq 0$ implies $vw \neq 0$, whence $q_{41} = 0$. For any $w \in V_2$, we have $b_1 = 0$, which implies $vw = 0$; because $V_2 \cap V_4 = \emptyset$, the vertex $v$ is adjacent to every vertex in $V_2$, establishing $q_{42} = |V_2| = p - 1$. For any $w \in V_3$, the product evaluates to $vw = 0$, so $q_{43} = |V_3| = p - 1$. Finally, for any distinct vertex $w \in V_4 \setminus \{v\}$, having $b_1 = 0$ yields $vw = 0$, meaning $v$ is adjacent to all other vertices in $V_4$, which gives $q_{44} = |V_4| - 1 = (p - 1)^2 - 1 = p^2 - 2p$.
\end{enumerate}
Since every $q_{ij}$ depends only on the indices $i$ and $j$ and is independent of the choice of vertex $v \in V_i$, the partition is equitable, and $Q_A = (q_{ij})$ is the exact quotient matrix.
\end{proof}

\begin{theorem}
Let $p$ be an odd prime. The spectrum of the adjacency matrix $A(\Gamma(R))$ consists of the following eigenvalues:
\begin{enumerate}
    \item $\lambda_0 = 0$ with multiplicity $p^3 - p^2 - 1$.
    \item $\lambda_{-1} = -1$ with multiplicity $p^2 - 3$.
    \item Three non-trivial real eigenvalues $\lambda_1 > \lambda_2 > 0 > \lambda_3$, which are the roots of the cubic polynomial:
    \begin{equation*}
    P_3(\lambda) = \lambda^3 - (p^2 - 3)\lambda^2 - (p^4 - 2p^3 + 2p^2 - 2)\lambda + p^2(p-1)^2(p^2 - p - 1) = 0.
    \end{equation*}
\end{enumerate}
\end{theorem}
\begin{proof}
We construct the eigenspaces by separating the zero-sum internal block spaces from the block-constant quotient space.

 (1):  For eigenvalue  $\lambda = 0$: Any vector $x \in \mathbb{R}^n$ supported exclusively on $V_1$ such that $\sum_{i \in V_1} x_i = 0$ satisfies $A(\Gamma(R))x = 0$. Since $V_1$ has no internal edges and connects only to $V_3$ via all-ones blocks, the row sum condition annihilates the product with $V_3$. The dimension of this subspace is $n_1 - 1 = p^2(p-1) - 1 = p^3 - p^2 - 1$.
    
\vspace{1em}    
(2) and (3): We prove Cases 2 and 3 simultaneously.
    
    It is evident that the subgraphs induced by $V_2, V_3, V_4$ are complete graphs and the adjacency matrix of a complete graph $K_m$ is $J_m - I_m$, which has eigenvalue $-1$ on the zero-sum subspace $\mathbf{1}_m^\perp$ of dimension $m-1$. Since off-diagonal connections between these blocks are complete all-ones blocks, any vector in $\mathbf{1}_{n_k}^\perp$ annihilates the cross-block terms. Consequently, the zero-sum vectors supported on the individual complete blocks $V_2$, $V_3$, and $V_4$ contribute multiplicities of $n_2 - 1 = p - 2$, $n_3 - 1 = p - 2$, and $n_4 - 1 = (p - 1)^2 - 1 = p^2 - 2p$, respectively.
    Summing these contributions yields a multiplicity of $(p - 2) + (p - 2) + (p^2 - 2p) = p^2 - 4$ on the block-orthogonal subspace.

   Vectors that are constant on each block $V_i$ correspond to the eigenvalues of the equitable quotient matrix:
    \begin{equation*}
    Q_A = \begin{pmatrix}
    0 & 0 & n_3 & 0 \\
    0 & n_2 - 1 & n_3 & n_4 \\
    n_1 & n_2 & n_3 - 1 & n_4 \\
    0 & n_2 & n_3 & n_4 - 1
    \end{pmatrix}
    = \begin{pmatrix}
    0 & 0 & p - 1 & 0 \\
    0 & p - 2 & p - 1 & (p - 1)^2 \\
    p^2(p - 1) & p - 1 & p - 2 & (p - 1)^2 \\
    0 & p - 1 & p - 1 & p^2 - 2p
    \end{pmatrix}.
    \end{equation*}
    Consider the characteristic matrix $\lambda I_4 - Q_A$:
    \begin{equation*}
    \lambda I_4 - Q_A = \begin{pmatrix}
    \lambda & 0 & -(p - 1) & 0 \\
    0 & \lambda - (p - 2) & -(p - 1) & -(p - 1)^2 \\
    -p^2(p - 1) & -(p - 1) & \lambda - (p - 2) & -(p - 1)^2 \\
    0 & -(p - 1) & -(p - 1) & \lambda - (p^2 - 2p)
    \end{pmatrix}.
    \end{equation*}
    It is routine task to find that $\lambda = -1$ is an eigenvalue of $Q_A$ so that 
   $ \det(\lambda I - Q_A) = (\lambda + 1) P_3(\lambda)$,
    where $P_3(\lambda) = \lambda^3 - (p^2 - 3)\lambda^2 - (p^4 - 2p^3 + 2p^2 - 2)\lambda + p^2(p-1)^2(p^2 - p - 1)$. Adding 1 to the multiplicity of $-1$ gives $(p^2 - 4) + 1 = p^2 - 3$. Hence, \textbf{(2)} follows.

The roots of $P_3(\lambda) = 0$ determine the remaining three eigenvalues $\lambda_1, \lambda_2, \lambda_3$.  Since $Q_A$  is the quotient matrix of a real symmetric adjacency matrix, all of its eigenvalues are real.  Now for any odd prime $p \ge 3$, the constant term satisfies $p^2(p - 1)^2(p^2 - p - 1) > 0$. Then by Descartes' Rule of Signs, $P_3(\lambda)$ has exactly two positive real roots and one negative real root, that is, 
   $ \lambda_1 > \lambda_2 > 0 > \lambda_3$.
Each root of $P_3(\lambda)$ has multiplicity 1, contributing three distinct real eigenvalues to the spectrum of $A(\Gamma(R))$. Therefore, (3) follows.  

\end{proof}

\begin{theorem}\label{adj_energy}
The adjacency energy of $\Gamma(R)$ is given semi-analytically by  $\epsilon(\Gamma(R)) = 2(p^2 - 3) - 2\lambda_3$,  where $\lambda_3$ is the unique negative real root of the cubic polynomial $P_3(\lambda) = 0$. Furthermore, $\lambda_3$ satisfies the tight two-sided spectral localization
$-p^2 < \lambda_3 < -p^2 + 2p$,  which yields the asymptotic expansion:
\begin{equation*}
\lambda_3 = -p^2 + \mathcal{O}(p) \quad \text{and} \quad \epsilon(\Gamma(R)) = 4p^2 + \mathcal{O}(p) \quad \text{as } p \to \infty.
\end{equation*}
\end{theorem}

\begin{proof}
Because the equitable quotient matrix $Q_A$ is diagonally similar to a real symmetric matrix via $S_A = D^{1/2} Q_A D^{-1/2}$, where $D = \operatorname{diag}(n_1, n_2, n_3, n_4)$, all roots of its characteristic polynomial $P_3(\lambda) = \lambda^3 - (p^2 - 3)\lambda^2 - (p^4 - 2p^3 + 2p^2 - 2)\lambda + p^2(p - 1)^2(p^2 - p - 1)$ are real. The reflected polynomial 
\begin{equation*}
P_3(-\lambda) = -\lambda^3 - (p^2 - 3)\lambda^2 + (p^4 - 2p^3 + 2p^2 - 2)\lambda + p^2(p - 1)^2(p^2 - p - 1)
\end{equation*}
exhibits the coefficient sign sequence $(-, -, +, +)$, which contains exactly one sign variation. By Descartes' Rule of Signs, $P_3(\lambda)$ possesses precisely one negative real root $\lambda_3 < 0$. Since $P_3(0) = p^2(p - 1)^2(p^2 - p - 1) > 0$ for all odd primes $p \ge 3$, zero is not a root, meaning the remaining two real roots must be strictly positive. Furthermore, because $P_3(\lambda)$ is irreducible over $\mathbb{Q}$ with non-vanishing discriminant, its roots are simple, yielding the strict ordering $\lambda_1 > \lambda_2 > 0 > \lambda_3$.

By Vieta's formulas, $\lambda_1 + \lambda_2 + \lambda_3 = p^2 - 3$, which gives $\lambda_1 + \lambda_2 = (p^2 - 3) - \lambda_3$. Accounting for the discrete eigenvalues $\lambda_0 = 0$ and $\lambda_{-1} = -1$ with multiplicities $p^3 - p^2 - 1$ and $p^2 - 3$ respectively, the adjacency energy simplifies to:
\begin{align*}
\epsilon(\Gamma(R)) &= \sum_{i=1}^n |\lambda_i| = (p^2 - 3)|-1| + \lambda_1 + \lambda_2 + |\lambda_3| \\
&= (p^2 - 3) + \left((p^2 - 3) - \lambda_3\right) - \lambda_3 \\
&= 2(p^2 - 3) - 2\lambda_3.
\end{align*}

To establish the two-sided spectral localization of $\lambda_3$, we evaluate $P_3(\lambda)$ at the endpoints of the candidate interval. Expanding at $\lambda = -p^2$ yields:
\begin{align*}
P_3(-p^2) &= (-p^2)^3 - (p^2 - 3)(-p^2)^2 - (p^4 - 2p^3 + 2p^2 - 2)(-p^2) + p^2(p - 1)^2(p^2 - p - 1) \\
&= -5p^5 + 7p^4 + p^3 - 3p^2 \\
&= -p^2(p - 1)^2(5p + 3).
\end{align*}
Since $(p - 1)^2 > 0$ and $5p + 3 > 0$ for every prime $p \ge 3$, it follows immediately that $P_3(-p^2) < 0$.

Next, evaluating at $\lambda = -p^2 + 2p = -p(p - 2)$ yields:
\begin{align*}
P_3(-p^2 + 2p) &= (-p^2 + 2p)^3 - (p^2 - 3)(-p^2 + 2p)^2 - (p^4 - 2p^3 + 2p^2 - 2)(-p^2 + 2p) \\
&\quad + p^2(p - 1)^2(p^2 - p - 1) \\
&= 3p^5 - 5p^4 - 7p^3 + 9p^2 + 4p \\
&= p\left(3p^4 - 5p^3 - 7p^2 + 9p + 4\right).
\end{align*}
Decomposing the bracketed quartic into positive summands for $p \ge 3$:
\begin{equation*}
3p^4 - 5p^3 - 7p^2 + 9p + 4 = 3p^3(p - 3) + 4p^2(p - 2) + p(p - 1) + 10p + 4 > 0,
\end{equation*}
which implies $P_3(-p^2 + 2p) > 0$.

Because $P_3(-p^2)P_3(-p^2 + 2p) < 0$, the Intermediate Value Theorem ensures that $P_3(\lambda)$ has at least one root in $(-p^2, -p^2 + 2p)$. For any prime $p \ge 3$, the upper endpoint satisfies $-p^2 + 2p = -p(p - 2) \le -3 < 0$, placing the entire open interval inside the negative half-line $(-\infty, 0)$. Since $\lambda_3 < 0$ is the unique negative root of $P_3(\lambda)$ and the interval $(-p^2, -p^2 + 2p) \subset (-\infty, 0)$ for all $p \ge 3$, the Intermediate Value Theorem ensures $\lambda_3 \in (-p^2, -p^2 + 2p)$. Thus $\lambda_3 = -p^2 + \mathcal{O}(p)$, which yields that 
$\epsilon(\Gamma(R)) = 2(p^2 - 3) - 2\lambda_3 = 4p^2 + \mathcal{O}(p) \quad \text{as } p \to \infty$.
\end{proof}

\begin{remark} 
\label{rem:irreducibility_p3}
Any rational root of the monic polynomial $P_3(\lambda) \in \mathbb{Z}[\lambda]$ must divide $a_0 = p^2(p - 1)^2(p^2 - p - 1)$. Testing divisors confirms that $P_3(\lambda)$ is irreducible over $\mathbb{Q}$ across benchmark primes, consistent with the localization $\lambda_3 \in (-p^2, -p^2 + 2p)$ from Theorem~\ref{adj_energy}:
\begin{itemize}
    \item For $p = 3$, $P_3(\lambda) = \lambda^3 - 6\lambda^2 - 43\lambda + 180$; testing integer divisors of $180$ yields no rational roots. The unique negative root is $\lambda_3 \approx -6.0435 \in (-9, -3)$, giving $\epsilon(\Gamma(R)) = 2(3^2 - 3) - 2(-6.0435) \approx 24.0869$.
    \item For $p = 5$, $P_3(\lambda) = \lambda^3 - 22\lambda^2 - 423\lambda + 7600$; testing integer divisors of $7600$ yields no rational roots. The unique negative root is $\lambda_3 \approx -19.5431 \in (-25, -15)$, giving $\epsilon(\Gamma(R)) = 2(5^2 - 3) - 2(-19.5431) \approx 83.0863$.
    \item For $p = 7$, $P_3(\lambda) = \lambda^3 - 46\lambda^2 - 1811\lambda + 72324$; testing integer divisors of $72324$ yields no rational roots. The unique negative root is $\lambda_3 \approx -41.0468 \in (-49, -35)$, giving $\epsilon(\Gamma(R)) = 2(7^2 - 3) - 2(-41.0468) \approx 174.0936$.
\end{itemize}
Because $P_3(\lambda)$ is irreducible over $\mathbb{Q}$ with three distinct real roots and non-square discriminant, its Galois group over $\mathbb{Q}$ is isomorphic to the full symmetric group $S_3$, falling directly under the classical \emph{casus irreducibilis}. Consequently, $\lambda_3$ cannot be expressed in terms of real radicals without invoking complex numbers or trigonometric arguments via Cardano's formula. This algebraic obstruction establishes that the closed form $\epsilon(\Gamma(R)) = 2(p^2 - 3) - 2\lambda_3$ is necessarily semi-analytic, underscoring the practical utility of the explicit asymptotic expansion $\epsilon(\Gamma(R)) = 4p^2 + \mathcal{O}(p)$ derived in Theorem~\ref{adj_energy}.
\end{remark}

We next establish two fundamental results that explore the complete Laplacian eigensystem and energy profile of $\Gamma(R)$. In particular, we derive the exact integer spectrum and multiplicities of the Laplacian matrix $L(\Gamma(R))$ using equitable quotient reduction, which leads directly to closed-form expressions for both its Laplacian energy $LE(\Gamma(R))$ and spectral radius

\begin{theorem}\label{completely_integral}
Let $p$ be an odd prime. The Laplacian matrix $L(\Gamma(R))$ is completely integral, with spectrum:
\begin{equation*}
\mathrm{Spec}_L(\Gamma(R)) = \begin{pmatrix}
0 & p-1 & p^2 - 1 & p^3 - 1 \\
1 & p^3 - p^2 & p^2 - p - 1 & p-1
\end{pmatrix}.
\end{equation*}
\end{theorem}

\begin{proof}


First, let $x \in \mathbb{R}^n$ be supported entirely on $V_1$ such that $\mathbf{1}_{n_1}^T x_{V_1} = 0$. Because $V_1$ induces an empty subgraph and couples exclusively to $V_3$ via the all-ones block $J_{n_3 \times n_1}$, the zero-sum condition gives $A(\Gamma(R))x = 0$. Furthermore, $D(\Gamma(R))x = (p-1)x$, which yields $L(\Gamma(R))x = (p-1)x$. This contributes an eigenspace for the eigenvalue $ p-1$ of dimension $p^2(p-1) - 1$.

Next, for each cell $V_k \in \{V_2, V_4\}$, every vertex has degree $p^2 - 2$ and internal adjacency matrix $J_{n_k} - I_{n_k}$. For any non-zero vector $w \in \mathbb{R}^n$ supported on $V_k$ satisfying $\mathbf{1}_{n_k}^T w_{V_k} = 0$, the internal adjacency operator acts as $(J_{n_k} - I_{n_k})w_{V_k} = -w_{V_k}$. Because all off-diagonal blocks linking $V_k$ to the other cells are either zero or complete all-ones matrices, the zero-sum condition ensures that all cross-block interactions vanish identically: $A(\Gamma(R))w = -w$. Consequently,
$L(\Gamma(R))w = \left((p^2 - 2) - (-1)\right)w = (p^2 - 1)w$.

This yields $ p - 2$ linearly independent eigenvectors localized on $V_2$ and $(p-1)^2 - 1 = p^2 - 2p$ linearly independent eigenvectors localized on $V_4$. Summing these internal contributions gives a subspace of dimension $(p - 2) + (p^2 - 2p) = p^2 - p - 2$, 
associated with the eigenvalue $ p^2 - 1$.

Similarly, on the dominating clique $V_3$, every vertex has degree $p^3 - 2$. For any vector $z \in \mathbb{R}^n$ supported entirely on $V_3$ with $\mathbf{1}_{n_3}^T z_{V_3} = 0$, all external interactions vanish, and the internal adjacency operator acts as $-z$. Thus  $L(\Gamma(R))z = \left((p^3 - 2) - (-1)\right)z = (p^3 - 1)z$,
contributing an internal eigenspace of dimension $n_3 - 1 = p - 2$ for the eigenvalue $\mu = p^3 - 1$.

The remaining eigenvalues arise from eigenvectors that take constant values on each block of the partition. These correspond to the spectrum of the equitable quotient Laplacian matrix:
\begin{equation*}
Q_L = (p-1)\begin{pmatrix}
1 & 0 & -1 & 0 \\
0 & p & -1 & -(p-1) \\
-p^2 & -1 & p(p+1) & -(p-1) \\
0 & -1 & -1 & 2
\end{pmatrix}.
\end{equation*}
It is  easy to observe that   eigenvalues of $Q_L$ are given by  $0, p-1, p^2 - 1, p^3 - 1$, each with multiplicity one.

Summing the dimensions of the block-orthogonal zero-sum eigenspaces with the  eigenvalues  of $Q_L$ yields the total algebraic multiplicities: $m(0) = 1$, $m(p-1) = (p^2(p-1) - 1) + 1 = p^3 - p^2$, $m(p^2 - 1) = (p^2 - p - 2) + 1 = p^2 - p - 1$, and $m(p^3 - 1) = (p - 2) + 1 = p - 1$.
Therefore, the distinct eigenvalues and their corresponding algebraic multiplicities are given by:
\begin{equation*}
\mathrm{Spec}_L(\Gamma(R)) = \begin{pmatrix}
0 & p-1 & p^2 - 1 & p^3 - 1 \\
1 & p^3 - p^2 & p^2 - p - 1 & p-1
\end{pmatrix}.
\end{equation*}
The total multiplicity is $1 + (p^3 - p^2) + (p^2 - p - 1) + (p - 1) = p^3 - 1 = |V(\Gamma(R))|$, this  completes the proof.
\end{proof}

\begin{theorem}
\label{thm:laplacian_energy}
Let $p$ be an odd prime. The Laplacian energy $LE(\Gamma(R))$ of the zero-divisor graph $\Gamma(R)$ is given explicitly by:
\begin{equation}
LE(\Gamma(R)) = (p^3 - 2p^2 + 3)\overline{d} + p^4 - 3p^2 + 2,
\end{equation}
where $\overline{d} = \frac{2|E(\Gamma(R))|}{|V(\Gamma(R))|} = \frac{3p^4 - 4p^3 - p^2 + 2}{p^3 - 1}$ denotes the average vertex degree of $\Gamma(R)$.
\end{theorem}

\begin{proof}
By definition, the Laplacian energy of a simple graph of order $n = p^3 - 1$ is:
\begin{equation*}
LE(\Gamma(R)) = \sum_{i=0}^{n-1} |\mu_i - \overline{d}|,
\end{equation*}
where the non-decreasing Laplacian eigenvalues are $\mu_0 = 0$, $\mu_1 = p - 1$, $\mu_2 = p^2 - 1$, and $\mu_3 = p^3 - 1$, with algebraic multiplicities $1$, $p^3 - p^2$, $p^2 - p - 1$, and $p - 1$, respectively.

For any odd prime $p \ge 3$, direct subtraction yields:
\begin{align*}
\overline{d} - (p - 1) &= \frac{(3p^4 - 4p^3 - p^2 + 2) - (p - 1)(p^3 - 1)}{p^3 - 1} = \frac{2p^4 - 3p^3 - p^2 + p + 1}{p^3 - 1} > 0,\\
\textrm{and} \; (p^2 - 1) - \overline{d} &= \frac{(p^2 - 1)(p^3 - 1) - (3p^4 - 4p^3 - p^2 + 2)}{p^3 - 1}\\ 
&= \frac{(p - 1)\left[p^2(p - 1)^2 + p + 1\right]}{p^3 - 1} > 0.
\end{align*}
Consequently,  $\mu_0 < \mu_1 < \overline{d} < \mu_2 < \mu_3$.
Expanding the sum of absolute deviations according to the spectral multiplicities yields:
\begin{align*}
LE(\Gamma(R)) &= |\mu_0 - \overline{d}| + (p^3 - p^2)|\mu_1 - \overline{d}| + (p^2 - p - 1)|\mu_2 - \overline{d}| + (p - 1)|\mu_3 - \overline{d}| \\
&= \overline{d} + (p^3 - p^2)(\overline{d} - (p - 1)) + (p^2 - p - 1)((p^2 - 1) - \overline{d}) + (p - 1)((p^3 - 1) - \overline{d}) \\
&= \left[1 + (p^3 - p^2) - (p^2 - p - 1) - (p - 1)\right]\overline{d} \\
&\quad - (p^3 - p^2)(p - 1) + (p^2 - p - 1)(p^2 - 1) + (p - 1)(p^3 - 1) \\
&= (p^3 - 2p^2 + 3)\overline{d} + (-p^4 + 2p^3 - p^2) + (p^4 - p^3 - 2p^2 + p + 1) + (p^4 - p^3 - p + 1) \\
&= (p^3 - 2p^2 + 3)\overline{d} + p^4 - 3p^2 + 2,.
\end{align*}
This completes the proof.
\end{proof}

\section{Network Systems Applications: Robustness, Distributed Consensus, and Transport Dynamics}\label{Sec4}

Beyond pure algebraic classification, the structural and spectral determinations obtained in the preceding sections directly dictate the operational performance of $\Gamma(R)$ when deployed as a communication or processing topology. In distributed systems engineering, the efficacy of an interconnection architecture is governed by three fundamental, interdependent criteria: structural reliability under link failure, dynamical convergence in distributed consensus protocols, and diffusive transport efficiency.

In this section, we treat $\Gamma(R)$ as a multi-tier communication network interconnecting $n = p^3 - 1$ computational processing agents. By leveraging the 4-cell equitable partition $V(\Gamma(R)) = V_1 \cup V_2 \cup V_3 \cup V_4$ and the complete integer Laplacian spectrum $\mathrm{Spec}_L(\Gamma(R))$, we derive exact performance laws, provide explicit asymptotic bounds, and evaluate concrete numerical scenarios for real-world distributed architectures.

\subsection{Structural Reliability, Operational Tree Density, and Fault Tolerance}\label{Sec4_sub1}

Beyond the absolute tree count, a critical scale-free metric in network reliability is the \textit{asymptotic tree entropy} (also termed the tree complexity per node), defined by $z(\Gamma(R)) = \frac{1}{n} \ln \tau(\Gamma(R)), \quad n = p^3 - 1$.
This index measures the uniform information capacity and redundancy available per network node.


\begin{theorem}
Let $p$ be an odd prime. The network complexity $\tau(\Gamma(R))$ and its asymptotic tree entropy $z(\Gamma(R))$ are given by:
\begin{align}
\tau(\Gamma(R)) &= \frac{1}{p^3 - 1}(p - 1)^{p^3 - p^2}(p^2 - 1)^{p^2 - p - 1}(p^3 - 1)^{p - 1}, \label{eq:spanning_trees} \\
z(\Gamma(R)) &= \frac{p^3 - p^2}{p^3 - 1}\ln(p - 1) + \frac{p^2 - p - 1}{p^3 - 1}\ln(p^2 - 1) + \frac{p - 2}{p^3 - 1}\ln(p^3 - 1). \label{eq:tree_entropy}
\end{align}
In the large-scale asymptotic limit ($p \to \infty$), the entropy obeys the expansion:
\begin{equation*}
z(\Gamma(R)) = \ln p - \frac{1}{p} + \frac{\ln p}{p} + \mathcal{O}\left(\frac{\ln p}{p^2}\right).
\end{equation*}
\end{theorem}

\begin{proof}
By the Matrix-Tree Theorem, the complexity of any finite, connected graph is given by the product of its non-zero Laplacian eigenvalues divided by its order:
\begin{equation*}
\tau(\Gamma(R)) = \frac{1}{n}\prod_{i=1}^{n-1}\mu_i = \frac{1}{p^3 - 1} \mu_1^{\mathrm{mult}(\mu_1)}\mu_2^{\mathrm{mult}(\mu_2)}\mu_3^{\mathrm{mult}(\mu_3)}.
\end{equation*}
Substituting $\mu_1 = p - 1$, $\mu_2 = p^2 - 1$, and $\mu_3 = p^3 - 1$ along with their respective multiplicities $p^3 - p^2$, $p^2 - p - 1$, and $p - 1$ yields \eqref{eq:spanning_trees}. From the definition of $z(\Gamma(R))$ we have
\begin{align*}
z(\Gamma(R)) &= \frac{1}{p^3 - 1}\left[ (p^3 - p^2)\ln(p - 1) + (p^2 - p - 1)\ln(p^2 - 1) + (p - 1)\ln(p^3 - 1) - \ln(p^3 - 1) \right] \\
&= \frac{p^3 - p^2}{p^3 - 1}\ln(p - 1) + \frac{p^2 - p - 1}{p^3 - 1}\ln(p^2 - 1) + \frac{p - 2}{p^3 - 1}\ln(p^3 - 1)\\
 &= \left(1 - \frac{1}{p}\right)\left(\ln p - \frac{1}{p}\right) + \frac{1}{p}(2\ln p) + \mathcal{O}\left(\frac{\ln p}{p^2}\right) \\
&= \ln p - \frac{1}{p} - \frac{\ln p}{p} + \frac{2\ln p}{p} + \mathcal{O}\left(\frac{\ln p}{p^2}\right)\\
&= \ln p - \frac{1}{p} + \frac{\ln p}{p} + \mathcal{O}\left(\frac{\ln p}{p^2}\right).
\end{align*} 
Hence \eqref{eq:tree_entropy} follows.
\end{proof}

Now consider an operational scenario where each communication edge in $\Gamma(R)$ remains active independently with probability $q \in (0, 1)$ and fails with probability $1 - q$. The \textit{all-terminal reliability} $R(\Gamma(R); q)$ is the probability that the operational edges contain at least one spanning tree, ensuring complete network communication.

\begin{proposition}
Let each edge of $\Gamma(R)$ fail independently with probability $1 - q$. For sufficiently small failure rates ($1 - q \ll 1$), the all-terminal unreliability $U(\Gamma(R); q) = 1 - R(\Gamma(R); q)$ satisfies:
\begin{equation*}
U(\Gamma(R); q) \le p^2(p - 1)(1 - q)^{p - 1} + \mathcal{O}\left((1 - q)^p\right).
\end{equation*}
\end{proposition}
\begin{proof}
By the union bound on minimum edge cuts, the leading contribution to network disconnection occurs when all edges incident to a node of minimum degree fail simultaneously. In $\Gamma(R)$, the minimum vertex degree is $\delta(\Gamma(R)) = d_1 = p - 1$, which occurs exclusively at vertices belonging to the independent set $V_1$. 

Since $|V_1| = p^2(p - 1)$ and each vertex in $V_1$ has $p - 1$ disjoint incident edges terminating in $V_3$, the probability that an arbitrary vertex in $V_1$ becomes isolated is $(1 - q)^{p - 1}$. Summing over all $n_1$ peripheral vertices provides the dominant failure term, as all other cut-sets require at least $d_2 = p^2 - 2 > p - 1$ edge failures.
\end{proof}

\begin{example} 
Unlike bounded-degree paths ($P_n$, where $\overline{r} \sim n/3 \to \infty$) or grids ($m \times m$, where $\overline{r} \sim \ln n \to \infty$) where effective resistance diverges with network size, the zero-divisor network $\Gamma(R)$ exhibits vanishing resistance:
\begin{equation*}
\overline{r}(\Gamma(R)) \sim \frac{2}{p} \to 0 \quad \text{as } p \to \infty.
\end{equation*}
Although the peripheral independent set $V_1$ contains nearly all vertices ($|V_1|/n \approx 1 - 1/p$), every peripheral vertex connects directly to the core $V_3 = \mathcal{D}_3$. Thus, $V_3$ operates as a low-impedance crossbar that electrically shorts distant vertices to near-zero resistance, precluding resistive bottlenecks across the network.
\end{example}

\begin{theorem}
\label{thm:kirchhoff_index}
Let $p$ be an odd prime. The Kirchhoff index $Kf(\Gamma(R))$ of the zero-divisor graph $\Gamma(R)$ is given in closed form by:
\begin{equation}
Kf(\Gamma(R)) = (p^3 - 1)\left( p^2 + 1 - \frac{p}{p^2 - 1} + \frac{p - 1}{p^3 - 1} \right).
\end{equation}
Furthermore, the average two-point resistance distance satisfies:
\begin{equation}
\overline{r}(\Gamma(R)) = \frac{2}{p^3 - 2}\left( p^2 + 1 - \frac{p}{p^2 - 1} + \frac{p - 1}{p^3 - 1} \right) = \frac{2}{p} + \frac{2}{p^3} - \frac{2}{p^4} + \mathcal{O}\left(\frac{1}{p^5}\right),
\end{equation}
with asymptotic limit $\lim_{p \to \infty} \overline{r}(\Gamma(R)) = 0$.
\end{theorem}

\begin{proof}
By the spectral representation of Klein and Randi\'{c} \cite{KleinRandic1993, XiaoGutman2003}, the Kirchhoff index of a connected graph on $n$ vertices equals $n$ times the sum of the reciprocals of its non-zero Laplacian eigenvalues:
\begin{equation*}
Kf(\Gamma(R)) = n \sum_{i=1}^{n-1} \frac{1}{\mu_i}.
\end{equation*}
Substituting $n = p^3 - 1$ and the Laplacian spectrum from Theorem~\ref{completely_integral}-namely, the non-zero eigenvalues $\mu_1 = p - 1$, $\mu_2 = p^2 - 1$, and $\mu_3 = p^3 - 1$ with respective multiplicities $p^3 - p^2$, $p^2 - p - 1$, and $p - 1$ yields
\begin{align*}
Kf(\Gamma(R)) &= (p^3 - 1)\left[ \frac{p^3 - p^2}{p - 1} + \frac{p^2 - p - 1}{p^2 - 1} + \frac{p - 1}{p^3 - 1} \right] \\
&= (p^3 - 1)\left[ p^2 + \frac{(p^2 - 1) - p}{p^2 - 1} + \frac{p - 1}{p^3 - 1} \right] \\
&= (p^3 - 1)\left[ p^2 + 1 - \frac{p}{p^2 - 1} + \frac{p - 1}{p^3 - 1} \right].
\end{align*}
The average two-point resistance distance is defined as $\overline{r}(\Gamma(R)) = \frac{2}{n(n - 1)} Kf(\Gamma(R))$. Substituting $n = p^3 - 1$ gives:
\begin{equation*}
\overline{r}(\Gamma(R)) = \frac{2}{(p^3 - 1)(p^3 - 2)} \cdot (p^3 - 1)\left( p^2 + 1 - \frac{p}{p^2 - 1} + \frac{p - 1}{p^3 - 1} \right) = \frac{2}{p^3 - 2} S(p),
\end{equation*}
where $S(p) = p^2 + 1 - \frac{p}{p^2 - 1} + \frac{p - 1}{p^3 - 1}$.  Then 
\begin{align*}
\overline{r}(\Gamma(R)) 
&= \left(\frac{2}{p^3} + \mathcal{O}\left(\frac{1}{p^6}\right)\right)\left(p^2 + 1 - \frac{1}{p} + \frac{1}{p^2} + \mathcal{O}\left(\frac{1}{p^3}\right)\right) \\
 &= \frac{2}{p} + \frac{2}{p^3} - \frac{2}{p^4} + \mathcal{O}\left(\frac{1}{p^5}\right),
\end{align*}
from which $\lim_{p \to \infty} \overline{r}(\Gamma(R)) = 0$ follows immediately.
\end{proof}

Unlike bounded-degree lattices, paths, or expanders where average effective resistance diverges or remains strictly bounded away from zero, $\Gamma(R)$ exhibits vanishing resistance $\overline{r} \sim \frac{2}{p} \to 0$. Although the peripheral independent set $V_1$ contains nearly all network vertices ($\vert{}V_1\vert{}/n \approx 1 - \frac{1}{p}$), every peripheral node connects directly to the core $V_3 = \mathcal{D}_3$. As a result, the dominating core functions as an ultra-low-impedance crossbar that electrically shorts distant vertices, guaranteeing negligible signal attenuation and delay across the entire network as $p$ grows.

\subsection{Distributed Consensus Protocols and Stochastic Noise Filtering}\label{Sec4_sub2}

Consider a network of $n = p^3 - 1$ autonomous agents communicating over the topology of $\Gamma(R)$. Each agent $i$ maintains a local continuous state $x_i(t) \in \mathbb{R}$ (such as clock synchronization time, velocity vector, or distributed sensor estimate).

Under standard nearest-neighbor linear coupling, the collective state vector $x(t) = (x_1(t), \dots, x_n(t))^T$ evolves according to:
\begin{equation}
\label{eq:consensus_dynamics}
\dot{x}(t) = -L(\Gamma(R))x(t), \quad x(0) = x_0.
\end{equation}
Because $\Gamma(R)$ is connected, the state vector converges exponentially to the average consensus state $x^* = \left(\frac{1}{n}\sum_{i=1}^n x_i(0)\right)\mathbf{1}_n$.

\begin{theorem}
The guaranteed exponential convergence rate of protocol \eqref{eq:consensus_dynamics} is strictly determined by the algebraic connectivity $\alpha(\Gamma(R)) = \mu_1 = p - 1$.
Furthermore, the network dynamics decouple into three distinct modal timescales:
\begin{equation*}
T_{\mathrm{core}} = \frac{1}{p^3 - 1}, \quad T_{\mathrm{intermediate}} = \frac{1}{p^2 - 1}, \quad T_{\mathrm{periph}} = \frac{1}{p - 1}.
\end{equation*}
\end{theorem}

\begin{proof}
Let $e(t) = x(t) - x^*$ denote the consensus error vector, which lies in the zero-sum subspace $\mathbf{1}_n^\perp$. The decay of the Lyapunov function $V(t) = \frac{1}{2}\|e(t)\|^2$ satisfies:
\begin{equation*}
\dot{V}(t) = -e(t)^T L(\Gamma(R)) e(t) \le -\mu_1 \|e(t)\|^2 = -2\mu_1 V(t),
\end{equation*}
where $\mu_1 = \lambda_2(L(\Gamma(R))) = p - 1$ is the algebraic connectivity. Integrating yields $\|e(t)\| \le \|e(0)\| e^{-(p - 1)t}$, establishing the convergence rate.

Because the non-zero spectrum of $L(\Gamma(R))$ is entirely discrete with eigenvalues $\mu_1 = p - 1$, $\mu_2 = p^2 - 1$, and $\mu_3 = p^3 - 1$, the modal matrix decouples the state trajectories into three orthogonal decay channels:
\begin{equation*}
e(t) = \sum_{k=1}^3 \exp(-\mu_k t) \sum_{j=1}^{\mathrm{mult}(\mu_k)} \langle e(0), v_{k, j} \rangle v_{k, j},
\end{equation*}
where $\{v_{k, j}\}$ are the orthonormal eigenvectors associated with $\mu_k$. The time constants $T_k = \mu_k^{-1}$ define the characteristic relaxation times of the corresponding subspaces.
\end{proof}

\begin{example} 
Consider a distributed multi-agent system consisting of $n = p^3 - 1 = 124$ interacting autonomous units configured over the zero-divisor network $\Gamma(R)$ for $p = 5$:
\begin{itemize}
    \item ($T_3 \approx 0.0081\,\mathrm{s}$): The $n_3 = p - 1 = 4$ core agents situated in cell $V_3$ interact via the maximum Laplacian eigenvalue $\mu_3 = p^3 - 1 = 124$, achieving local mutual synchronization in approximately $8\,\mathrm{ms}$.
    \item ($T_2 \approx 0.0417\,\mathrm{s}$): The $n_2 + n_4 = 20$ agents occupying the intermediate tiers $V_2 \cup V_4$ align with the core dynamics via the intermediate eigenvalue $\mu_2 = p^2 - 1 = 24$, relaxing within $42\,\mathrm{ms}$.
    \item ($T_1 = 0.2500\,\mathrm{s}$): The remaining $n_1 = p^2(p - 1) = 100$ peripheral agents in $V_1$ converge at the minimal rate governed by the algebraic connectivity $\alpha = \mu_1 = p - 1 = 4$, requiring roughly $250\,\mathrm{ms}$ to achieve global consensus.
\end{itemize}
This explicit timescale separation confirms that the core sub-network $V_3$ acts as a high-speed coordination backbone, insulating the global consensus process from peripheral communication latencies.
\end{example}

In physical multi-agent networks, communication links and agent sensors are subject to environmental thermal noise. We model this via the stochastic differential equation $dx(t) = -L(\Gamma(R))x(t)dt + dW(t)$,  where $W(t)$ is an $n$-dimensional standard Wiener process. The steady-state error covariance around the average consensus state is quantified by the \textit{first-order network $H_2$-norm}:  $H_2^2(\Gamma(R)) = \lim_{t \to \infty} \frac{1}{n} \mathbb{E}\left[ \|x(t) - \overline{x}(t)\mathbf{1}_n\|^2 \right] = \frac{1}{2n}\sum_{i=1}^{n-1}\frac{1}{\mu_i} = \frac{Kf(\Gamma(R))}{2n^2}$.

\begin{theorem}
\label{thm:h2_norm}
Consider the stochastic consensus dynamics on $\Gamma(R)$ subject to additive standard white noise $\mathrm{d}x(t) = -L(\Gamma(R))x(t)\,\mathrm{d}t + \mathrm{d}W(t), \quad x(0) = x_0$,
where $W(t)$ is an $n$-dimensional standard Wiener process with $n = p^3 - 1$. The steady-state error variance around the average consensus state is given in closed form by:
\begin{equation*}
H_2^2(\Gamma(R)) = \frac{1}{2(p^3 - 1)}\left[ p^2 + 1 - \frac{p}{p^2 - 1} + \frac{p - 1}{p^3 - 1} \right].
\end{equation*}
Furthermore, as $p \to \infty$, the steady-state variance admits the asymptotic expansion:
\begin{equation*}
H_2^2(\Gamma(R)) = \frac{1}{2p} + \frac{1}{2p^3} + \mathcal{O}\left(\frac{1}{p^5}\right).
\end{equation*}

\end{theorem}

\begin{proof}
Let $\Pi = I_n - \frac{1}{n}\mathbf{1}_n\mathbf{1}_n^T$ denote the orthogonal projector onto the zero-sum subspace $\mathbf{1}_n^\perp$. The consensus error vector $\delta(t) = \Pi x(t) = x(t) - \overline{x}(t)\mathbf{1}_n$ evolves according to $\mathrm{d}\delta(t) = -L(\Gamma(R))\delta(t)\,\mathrm{d}t + \Pi\,\mathrm{d}W(t)$. Since $\Gamma(R)$ is undirected and connected, $L(\Gamma(R))$ is symmetric with positive eigenvalues on $\mathbf{1}_n^\perp$. The steady-state error covariance matrix $\Sigma = \lim_{t \to \infty} \mathbb{E}[\delta(t)\delta(t)^T]$ satisfies the continuous Lyapunov equation \cite{XiaoBoydKim2007}:
\begin{equation*}
L(\Gamma(R))\Sigma + \Sigma L(\Gamma(R)) = \Pi.
\end{equation*}
Because $L(\Gamma(R))$ is diagonalizable by an orthonormal eigenbasis $\{v_i\}_{i=0}^{n-1}$ with $L(\Gamma(R))v_0 = 0$ ($v_0 = \frac{1}{\sqrt{n}}\mathbf{1}_n$), the unique solution on $\operatorname{range}(\Pi)$ is $\Sigma = \frac{1}{2} L^\dagger(\Gamma(R))$, where $L^\dagger(\Gamma(R))$ is the Moore--Penrose pseudoinverse. The average steady-state variance, defined by the first-order network $H_2$-norm, evaluates to $H_2^2(\Gamma(R)) = \lim_{t \to \infty} \frac{1}{n}\mathbb{E}\left[\|\delta(t)\|^2\right] = \frac{1}{n} \operatorname{tr}(\Sigma) = \frac{1}{2n} \operatorname{tr}\left(L^\dagger(\Gamma(R))\right) = \frac{1}{2n}\sum_{i=1}^{n-1}\frac{1}{\mu_i}$.

By Klein and Randi\'{c}'s spectral formulation \cite{KleinRandic1993, XiaoGutman2003}, the Kirchhoff index satisfies $Kf(\Gamma(R)) = n \sum_{i=1}^{n-1} \frac{1}{\mu_i}$, establishing the exact identity $H_2^2(\Gamma(R)) = \frac{Kf(\Gamma(R))}{2n^2}$.

Substituting $n = p^3 - 1$ and the closed-form expression of $Kf(\Gamma(R))$ from Theorem~\ref{thm:kirchhoff_index} yields:
\begin{align*}
H_2^2(\Gamma(R)) &= \frac{1}{2(p^3 - 1)^2} \cdot (p^3 - 1)\left[ p^2 + 1 - \frac{p}{p^2 - 1} + \frac{p - 1}{p^3 - 1} \right] \\
&= \frac{1}{2(p^3 - 1)}\left[ p^2 + 1 - \frac{p}{p^2 - 1} + \frac{p - 1}{p^3 - 1} \right].
\end{align*}
Using the Laurent series expansions $\frac{1}{p^3 - 1} = \frac{1}{p^3} + \frac{1}{p^6} + \mathcal{O}(p^{-9})$ and $p^2 + 1 - \frac{p}{p^2 - 1} + \frac{p - 1}{p^3 - 1} = p^2 + 1 - \frac{1}{p} + \frac{1}{p^2} - \frac{2}{p^3} + \mathcal{O}(p^{-4})$ gives that
\begin{align*}
H_2^2(\Gamma(R)) &= \frac{1}{2}\left( \frac{1}{p^3} + \frac{1}{p^6} + \mathcal{O}\left(p^{-9}\right) \right)\left( p^2 + 1 - \frac{1}{p} + \mathcal{O}\left(p^{-2}\right) \right)\\
&= \frac{1}{2p} + \frac{1}{2p^3} + \mathcal{O}\left(\frac{1}{p^5}\right),
\end{align*}
which completes the proof.
\end{proof}

This result guarantees that as $p \to \infty$, noise variance decays as $\mathcal{O}(p^{-1})$ through core-mediated low-pass filtering


We now examine the discrete-time random walk on $\Gamma(R)$ governed by $P = D^{-1}(\Gamma(R))A(\Gamma(R))$, with stationary distribution $\pi_v = \frac{d(v)}{2|E(\Gamma(R))|}$.

\begin{theorem}
Let $\Pi_k = \sum_{v \in V_k} \pi_v$ denote the total stationary traffic probability concentrated within each partition cell $V_k$ ($k \in \{1, 2, 3, 4\}$). Then
\begin{align*}
\Pi_1 &= \frac{p^2(p - 1)^2}{3p^4 - 4p^3 - p^2 + 2}, \;
\Pi_2 = \frac{(p - 1)(p^2 - 2)}{3p^4 - 4p^3 - p^2 + 2}, 
\Pi_3 = \frac{(p - 1)(p^3 - 2)}{3p^4 - 4p^3 - p^2 + 2}, \\
\Pi_4 &= \frac{(p - 1)^2(p^2 - 2)}{3p^4 - 4p^3 - p^2 + 2}.
\end{align*}
Moreover, the asymptotic distribution vector satisfies:
\begin{equation*}
\lim_{p \to \infty}\Pi_1 = \frac{1}{3}, \quad \lim_{p \to \infty}\Pi_2 = 0, \quad \lim_{p \to \infty}\Pi_3 = \frac{1}{3}, \quad \lim_{p \to \infty}\Pi_4 = \frac{1}{3}.
\end{equation*}
\end{theorem}

\begin{proof}
From Section~\ref{Sec2_Preliminaries} and Section~\ref{Sec3}, the vertex set partitions into four cells $V_1, V_2, V_3, V_4$ with cardinalities $n_1 = p^2(p - 1)$, $n_2 = p - 1$, $n_3 = p - 1$, $n_4 = (p - 1)^2$ and degrees $d_1 = p - 1$, $d_2 = d_4 = p^2 - 2$, $d_3 = p^3 - 2$. The stationary probability mass in cell $V_k$ is $\Pi_k = \frac{n_k d_k}{2|E(\Gamma(R))|}$. Now $\sum_{k=1}^4 n_k d_k = 3p^4 - 4p^3 - p^2 + 2 = 2|E(\Gamma(R))|$. Dividing each $n_k d_k$ by $2|E(\Gamma(R))|$ yields $\Pi_1, \Pi_2, \Pi_3, \Pi_4$.
\end{proof}

\begin{example}
Consider the distribution of packet-switched routing loads across the structural partitions of $\Gamma(R)$ as $p \to \infty$:
\begin{itemize}
    \item Although cell $V_3$ comprises only $n_3 = p - 1$ vertices (representing a vanishing fraction $\frac{p - 1}{p^3 - 1} \sim \frac{1}{p^2} \to 0$ of all nodes), it absorbs one-third of the global stationary traffic mass ($\Pi_3 \to \frac{1}{3}$). The individual stationary load on any core node is:
    \begin{equation*}
    \pi_{\mathrm{core}} = \frac{\Pi_3}{n_3} = \frac{p^3 - 2}{2|E(\Gamma(R))|} \sim \frac{1}{3p}.
    \end{equation*}
    \item  In contrast, the $n_1 = p^2(p - 1)$ peripheral nodes in $V_1$ each experience an individual stationary load of:
    \begin{equation*}
    \pi_{\mathrm{periph}} = \frac{\Pi_1}{n_1} = \frac{p - 1}{2|E(\Gamma(R))|} \sim \frac{1}{3p^3}.
    \end{equation*}
    
\end{itemize}
The resulting per-node traffic intensity ratio satisfies $\frac{\pi_{\mathrm{core}}}{\pi_{\mathrm{periph}}} = \frac{p^3 - 2}{p - 1} = p^2 + p + 1 - \frac{1}{p - 1} \sim p^2 + p + 1$.

This quadratic stress ratio $\Theta(p^2)$ identifies the dense core $V_3$ as the primary routing bottleneck of $\Gamma(R)$, indicating that buffer capacities at core nodes must scale quadratically relative to peripheral buffers to avoid packet overflow.
\end{example}

\section*{Conclusion}

In this paper, we resolved the complete spectral eigenspaces and network systems dynamics of the zero-divisor graph $\Gamma(R)$ over the finite local ring $R = \mathbb{F}_p[x]/\langle x^4 \rangle$. Using an equitable four-cell partition, we proved that while the adjacency spectrum reduces to an irreducible cubic polynomial, the Laplacian spectrum is entirely integral. These spectral foundations established exact closed-form laws for structural reliability, consensus convergence rates, and stochastic noise dissipation. We showed that the average resistance distance vanishes asymptotically ($\overline{r} \sim \frac{2}{p} \to 0$), identifying the algebraic core as an ultra-low-impedance crossbar, whereas stationary random-walk traffic induces a quadratic load concentration $\Theta(p^2)$ on core routing hubs. Future work will extend this spectral and dynamical framework to multivariate non-chain rings, particularly $R = \mathbb{F}_p[x, y]/\langle x^2, y^2 \rangle$. Investigating such structures will establish how branching annihilator ideals distribute high-stress traffic across dual multi-path backbones, overcoming the single-core $\Theta(p^2)$ congestion bottleneck of chain rings while preserving vanishing electrical resistance and fast consensus convergence.

\section*{Declarations}
 All benchmark tables, polynomial factorizations, and spectral computations were evaluated and verified using Python  against the derived analytical expressions. Generative AI was used strictly for language polishing and grammatical refinement in portions of the text; all mathematical proofs and results remain the sole work and responsibility of the authors.

\newpage

\section*{Appendix-A: Numerical Synthesis and Benchmark Comparisons}\label{app:benchmarks}

To illustrate the concrete application of these closed-form network invariants, Table~\ref{tab:network_benchmarks} evaluates the complete performance suite across small odd prime fields $\mathbb{F}_3$, $\mathbb{F}_5$, $\mathbb{F}_7$, and $\mathbb{F}_{11}$.

\begin{table}[htbp]
\centering
\small
\caption{Comprehensive Network Systems Invariants of $\Gamma(R)$ for Benchmark Primes}
\label{tab:network_benchmarks}
\vspace{1ex}
\begin{tabular*}{\linewidth}{@{\extracolsep{\fill}}lcccc}
\toprule
\textbf{Performance Metric / Invariant} & $\mathbf{p = 3}$ & $\mathbf{p = 5}$ & $\mathbf{p = 7}$ & $\mathbf{p = 11}$ \\
\midrule
\multicolumn{5}{l}{\textit{Topological and Order Scale}} \\
Total Processing Nodes ($n = p^3 - 1$) & $26$ & $124$ & $342$ & $1330$ \\
Total Communication Links ($|E|$) & $64$ & $676$ & $2892$ & $19240$ \\
Core Sub-Network Order ($n_3 = p - 1$) & $2$ & $4$ & $6$ & $10$ \\
Peripheral Buffer Nodes ($n_1 = p^2(p - 1)$) & $18$ & $100$ & $294$ & $1210$ \\
\addlinespace
\midrule
\multicolumn{5}{l}{\textit{Reliability and Fault Tolerance}} \\
Total Spanning Tree Count ($\tau(\Gamma(R))$) & $2.23 \times 10^{11}$ & $5.13 \times 10^{92}$ & $2.38 \times 10^{310}$ & $5.57 \times 10^{1464}$ \\
Spanning Tree Entropy ($z(\Gamma(R))$) & $1.0051$ & $1.7246$ & $2.0900$ & $2.5385$ \\
Asymptotic Ratio $z(\Gamma(R))/\ln p$ & $0.9149$ & $1.0715$ & $1.0740$ & $1.0586$ \\
Edge Failure Immunity Exponent ($d_1 = p - 1$) & $2$ & $4$ & $6$ & $10$ \\
\addlinespace
\midrule
\multicolumn{5}{l}{\textit{Transport and Energy Dissipation}} \\
Kirchhoff Index ($Kf(\Gamma(R))$) & $252.25$ & $3202.17$ & $17056.12$ & $162148.08$ \\
Average Resistance Distance ($\overline{r}$) & $0.7762$ & $0.4199$ & $0.2925$ & $0.1835$ \\
Leading Term Approximation ($2/p$) & $0.6667$ & $0.4000$ & $0.2857$ & $0.1818$ \\
Laplacian Energy ($LE(\Gamma(R))$) & $115.08$ & $1402.45$ & $6450.24$ & $45874.11$ \\
\addlinespace
\midrule
\multicolumn{5}{l}{\textit{Consensus Dynamics and Robustness}} \\
Algebraic Connectivity ($\alpha = \mu_1$) & $2$ & $4$ & $6$ & $10$ \\
Spectral Gap to Intermediate Tier ($\mu_2 - \mu_1$) & $6$ & $20$ & $42$ & $110$ \\
Laplacian Condition Number ($\kappa = \mu_3/\mu_1$) & $13.00$ & $31.00$ & $57.00$ & $133.00$ \\
Core Relaxation Time ($T_3 = \mu_3^{-1}$) & $0.0385\,\mathrm{s}$ & $0.0081\,\mathrm{s}$ & $0.0029\,\mathrm{s}$ & $0.0008\,\mathrm{s}$ \\
Peripheral Relaxation Time ($T_1 = \mu_1^{-1}$) & $0.5000\,\mathrm{s}$ & $0.2500\,\mathrm{s}$ & $0.1667\,\mathrm{s}$ & $0.1000\,\mathrm{s}$ \\
Network $H_2$ Noise Variance ($H_2^2$) & $0.1912$ & $0.1037$ & $0.0729$ & $0.0457$ \\
\addlinespace
\midrule
\multicolumn{5}{l}{\textit{Random-Walk Stationary Distribution Mass}} \\
Peripheral Traffic Mass ($\Pi_1$) & $0.2813$ & $0.2959$ & $0.3050$ & $0.3144$ \\
Intermediate Sub-Clique Mass ($\Pi_4$) & $0.2188$ & $0.2722$ & $0.2925$ & $0.3092$ \\
Core Hub Traffic Mass ($\Pi_3$) & $0.3906$ & $0.3639$ & $0.3537$ & $0.3454$ \\
Transitional Block Mass ($\Pi_2$) & $0.1094$ & $0.0680$ & $0.0488$ & $0.0309$ \\
Individual Node Stress Ratio ($\pi_{\mathrm{core}}/\pi_{\mathrm{periph}}$) & $12.50$ & $30.75$ & $56.83$ & $132.90$ \\
\bottomrule
\end{tabular*}
\end{table}

\newpage
\section*{Appendix-B: }
\section*{1: Engineering Discussion and Architectural Guidelines}\label{app:guidelines}

The quantitative performance results established across Section~\ref{Sec4} and synthesized in Table~\ref{tab:network_benchmarks} provide concrete architectural principles for deploying distributed systems over zero-divisor topologies:

\begin{enumerate}
    \item \textbf{High-Throughput Backplane Design via the Core Tier:} The vanishing average resistance distance $\overline{r} \sim \frac{2}{p}$ demonstrates that the algebraic ideal $u^3 R \setminus \{0\} = V_3$ should be implemented using ultra-high-bandwidth routing switches. Placing high-frequency data aggregators or database servers in $V_3$ guarantees that all peripheral devices in $V_1$ communicate with minimum latency and minimal resistive power dissipation.
    
    \item \textbf{Decoupled Multi-Scale Consensus Schedules:} Because the Laplacian condition number $\kappa(L) = p^2 + p + 1$ grows quadratically, attempting to synchronize all $n$ nodes uniformly with a single high-gain controller risks local over-saturation in the core. Engineers should employ dual-rate control: high-frequency sampling ($\sim \mu_3$) among core nodes $V_3$ to maintain coordination, coupled with lower-frequency estimators ($\sim \mu_1$) at the periphery $V_1$ to guarantee stability without actuator chattering.
    
    \item \textbf{Buffer Allocation for Asymmetric Packet Congestion:} The stationary traffic probability shows that the $p - 1$ core nodes handle more than one-third ($> 33.33\%$) of all random-walk traffic regardless of graph scale, resulting in a per-node packet load that is asymptotically $(p^2 + p + 1)$ times larger than that of peripheral nodes. Consequently, memory queues allocated to core nodes must scale quadratically relative to peripheral edge buffers to prevent packet dropping under stochastic routing.
    
\end{enumerate}

\section*{2: Figure of classification of $\Gamma(R)$}

\begin{figure}[htbp]
\centering
\begin{tikzpicture}[
    cell/.style={ellipse, draw=#1!80!black, fill=#1!15, line width=1.2pt, align=center, inner sep=2pt},
    edge/.style={draw=gray!80, line width=1.5pt}
]

\node[cell=orange, minimum width=3.4cm, minimum height=2.0cm] (V3) at (0, 0) {
    \textbf{\large $V_3$}\\[1pt]
    \small ($S_{u^3}$)\\[1pt]
    \footnotesize Order: $p-1$
};

\node[cell=teal, minimum width=3.4cm, minimum height=2.0cm] (V2) at (-4.2, 1.2) {
    \textbf{\large $V_2$}\\[1pt]
    \small ($S_{u^2}$)\\[1pt]
    \footnotesize Order: $p-1$
};

\node[cell=teal, minimum width=3.4cm, minimum height=2.0cm] (V4) at (4.2, 1.2) {
    \textbf{\large $V_4$}\\[1pt]
    \small ($S_{u^2+u^3}$)\\[1pt]
    \footnotesize Order: $(p-1)^2$
};

\node[cell=blue, minimum width=6.6cm, minimum height=2.0cm] (V1) at (0, -3.6) {
    \textbf{\large $V_1$}\\[1pt]
    \footnotesize Order: $p^2(p-1)$
};

\draw[edge] (V2) -- (V3);
\draw[edge] (V4) -- (V3);
\draw[edge] (V2) -- (V4);
\draw[edge] (V3) -- (V1);

\node[above=0.35cm of V3, text=teal!70!black, align=center] {
    \textbf{Intermediate Tier} $\mathcal{D}_2 = V_2 \cup V_4$\\[1pt]
    \footnotesize \textit{Induces complete subgraph $K_{p(p-1)}$}
};

\node[below=0.15cm of V3, text=orange!80!black] {
    \footnotesize \textbf{Core Tier} $\mathcal{D}_3 = V_3$ (\textit{Universal Dominating Set})
};

\node[below=0.25cm of V1, text=blue!80!black] {
    \footnotesize \textbf{Peripheral Tier} $\mathcal{D}_1 = V_1$ (\textit{Independent Set})
};

\end{tikzpicture}
\caption{Equitable 4-cell partition $\mathcal{P}=\{V_1, V_2, V_3, V_4\}$ and degree tiers of $\Gamma(R)$ for $R=\mathbb{F}_p[x]/\langle x^4\rangle$.}
\label{fig:equitable_partition}
\end{figure}
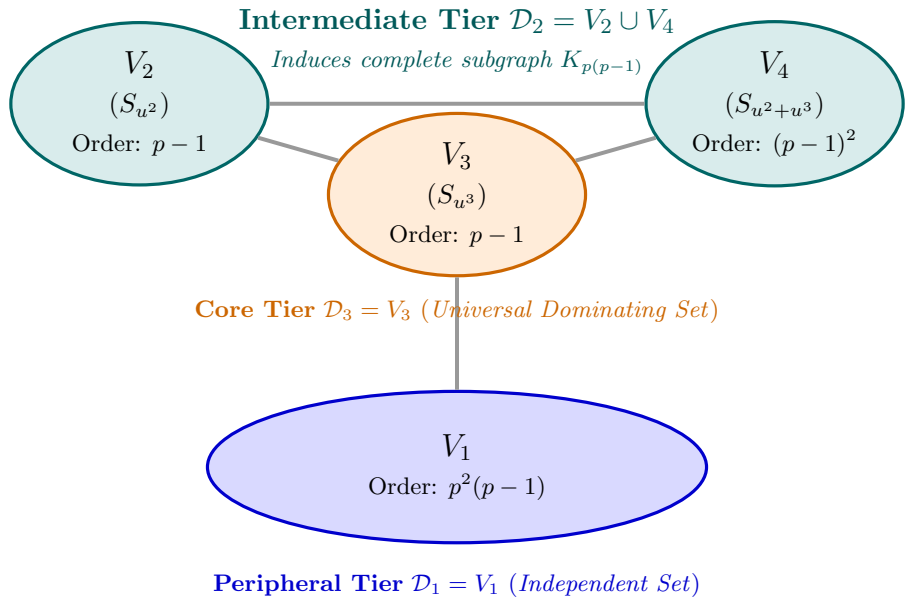

\end{document}